\documentclass[12pt]{article}
\usepackage{makeidx}
\usepackage{graphicx}
\usepackage{amsmath}
\usepackage{amsfonts}
\usepackage{hyperref}  
\usepackage[numbers]{natbib} 
\usepackage{enumitem} 
\usepackage{pdflscape}
\usepackage{amssymb}
\usepackage{amsthm}
\usepackage{float} 
 \usepackage{fancyhdr}  

 \usepackage[nameinlink,noabbrev]{cleveref}

  \theoremstyle{plain}

  \theoremstyle{definition}

  \newtheorem{proposition}{Proposition}

  \crefname{equation}{}{}
  
  \newfloat{textbox}{htbp}{lop}
  \floatname{textbox}{Text Box}

 \newcommand{\slot}{\resizebox{\width}{0.5\height}{$\sqcup$}}
   
 \newcommand{\Cen}[1]{Z(#1)}

\renewcommand{\Vec}[1]{\, \overrightarrow{#1}}

\title{Stickel-type key exchange with hidden subspaces}

\author{Fintan Costello \\ School of Computer Science and Informatics,\\ University College Dublin\\
	and \\
	Paul Watts \\ Department of Physics,\\National University of
	Ireland Maynooth\\ }

\date{}

\begin{document}
  
\maketitle \thispagestyle{empty}

  \begin{abstract}
  	We give a witness-finding cryptanalysis of Stickel-type key exchange schemes, which involve two-sided multiplication of $n \times n$ matrices over $\mathbb{F}_p$, where these matrices are drawn from public  subspaces with a particular commuting structure.   This analysis covers Stickel's original proposal \cite{stickel2005new}, Shpilrain's polynomial extension of that scheme \cite{shpilrain2005new},  Nager's algebraic extension of that scheme \cite{Nager2024}, and more generally all Stickel-type approaches using public subspaces over matrix algebra in finite fields: all such schemes can be broken in polynomial time.  We also describe a new key establishment scheme using two-sided matrix multiplication in which the commuting subspaces used to form the key are hidden via conjugation by private terms, blocking this specific public-subspace analysis; the witness-finding problem in this new scheme has a direct reduction from a standard NP-hard problem (Edmonds' problem).    
  \end{abstract}

Traditional approaches to cryptographic key exchange are vulnerable to quantum attacks; given this, interest has grown in alternative key exchange mechanisms.  One family of mechanism involves two-sided multiplication of  elements from a public subgroup,  subspace or ring over standard matrix algebra, originally proposed by Stickel \cite{stickel2005new} in terms of  powers of invertible $n \times n$ matrices over a finite field $\mathbb{F}_p$.  A range of key exchange schemes have been given using the general Stickel/two-sided-multiplication-by-public-subspace-elements approach;  we describe various forms of this approach and give a tensor-based cryptanalysis that applies to all such schemes, at least over classical matrix algebra.  We also describe a new hidden-subspace version of this approach, which removes the public linear structure used in this analysis.   We begin with  notation.

\section{Notation}

For some fixed prime $p$ we take $\mathcal{M}_k =\mathcal{M}_k(\mathbb{F}_p)$ to be the set of all $k \times k$ matrices with entries in  $\mathbb{F}_p$, and for some fixed dimension parameter $n$ take $\mathcal{M} =\mathcal{M}_n$; we assume that all matrices are members of $\mathcal{M}$ unless stated otherwise.  For any set of matrices $\mathcal{A}$ we take  $\mathcal{A}^\times$ to be the set of invertible matrices in $\mathcal{A}$. 
For a matrix $x$ we take $Alg(x)$ to be the polynomial algebra in $x$; for a matrix $x$ or set of matrices $\mathcal{X}$ we take
\[ \Cen{x} = \{ y \mid xy=yx\}, \quad \Cen{\mathcal{X}} = \{ y \mid  xy=yx \text{ for all } x \in \mathcal{X} \} \]
to be the set of matrices that commute with $x$ or with all members of $\mathcal{X}$.  We also take 
\[ \mathcal{Z}_k =\Cen{\mathcal{M}_k } = \{ \lambda I \mid \lambda \in \mathbb{F}_p \}  \]
 to represent the set of scalar matrices in $\mathcal{M}_k$ (the center of $\mathcal{M}_k$), and take $\mathcal{Z}  = \mathcal{Z}_n$.

For   matrices $a,b,W$ we define a two-sided action operator  $\slot$ such that 
\[ (a \slot b)(W) = aWb \] 
For  sets of matrices $\mathcal{A},\mathcal{B}$ we define  
\[ \mathcal{A} \slot \mathcal{B}  =  \{ a \slot b \mid a \in \mathcal{A}, b \in\mathcal{B} \}\] 
to be  the set of all such operators  generated by members of those sets; for matrix subspaces  $\mathcal{A},\mathcal{B}$ we also define
\[ \langle \mathcal{A} \slot \mathcal{B} \rangle =  \left\{\sum_{i=1}^{m}  \lambda_i u_i \mid u_i \in \mathcal{A} \slot \mathcal{B}, m > 0, \lambda_i \in \mathbb{F}_p \right\} \]
to be the set of all linear combinations of operators $ u \in \mathcal{A} \slot \mathcal{B} $. We take
\[  \langle \mathcal{A} W \mathcal{B} \rangle = \{ u(W) \mid u \in \langle \mathcal{A} \slot \mathcal{B} \rangle \}\] 
to be the set of all matrices produced by applying  these operators to $W$.

Using the vectorisation operator $\Vec{X}$ and the  Kronecker or tensor product $\otimes$ we have
\begin{equation}\label{eq:kronecker} (b^\top \otimes a) \Vec{W} =  \Vec{aWb} \end{equation}
(a standard identity) and so
\[  \Vec{(a \slot b)(W)} =  (b^\top \otimes a) \Vec{W} \]
and $(a \slot b)$ and  $(b^\top \otimes a)$ are equivalent modulo reshaping of their inputs and outputs; since linear combinations are equivalent modulo matching reshaping of their terms, every operator $u \in \langle \mathcal{A} \slot \mathcal{B} \rangle$ has an equivalent matrix $F_u$ (a sum of tensor products) such that
\[  \Vec{u(W)} = F_u \Vec{W} \]

For matrix subspaces and operators such that
\[  u = (a \slot b) \in \mathcal{A} \slot \mathcal{B}, \quad \quad v = (c \slot d)  \in \mathcal{C} \slot \mathcal{D}, \quad \quad \mathcal{C} \subseteq \Cen{\mathcal{A}}, \quad \mathcal{D} \subseteq \Cen{\mathcal{B}}  \]
we have
\[  u(v(W)) = a (c W d) b = c a W b d = c (a W b) d= v(u(W)) \]
and so all such operators commute; since linear combinations of commuting operators will also commute, for such commuting subspaces all operators
\[ u \in \langle \mathcal{A} \slot \mathcal{B} \rangle, \quad \quad v   \in \langle \mathcal{C} \slot \mathcal{D} \rangle \]
will also commute.

\section{Stickel-type key exchange schemes}

Stickel's original scheme \cite{stickel2005new} begins with an initialisation stage, where Alice and Bob agree public non-commuting and invertible  matrices $A,B$ and $W$ (the generators for the scheme). Given these generators, to construct a shared key $K$ Alice and Bob then exchange matrices as follows.  Alice picks  random integers $f$ and $g$, sets $a=A^{f}$ and $b = B^{g}$ and publishes the matrix
\[ U=  a W b\] 
while Bob  similarly picks  random integers $\ell$ and $m$, sets $c=A^{\ell}$ and $d = B^{m}$, and publishes  
\[ V =c  W d \]
Given $V$ Alice takes 
\[ a  V b = a c W d b = a c W b d   \]
and given $U$ Bob takes
\[ c  U d = c a W b d=  a c W b d  \]
(since powers of a given matrix commute) and Alice and Bob have a shared key $K =  a c W b d $.   In our notation this corresponds to defining  $\mathcal{A},\mathcal{C}$ to be the sets of all powers of $A$ and defining $\mathcal{B},\mathcal{D}$ to be the sets of all powers of $B$; then Alice and Bob pick
\[ u \in \mathcal{A} \slot \mathcal{B}, \quad v \in \mathcal{C} \slot \mathcal{D}  \]
and publish 
\[ U=u(W),  \quad  V=v(W) \]
and since these operator sets commute they have
\[ u(V)=  u(v(W))= v(u(W))  = v(U) =K \]

Note that the structure of this scheme means that the shared key $K$ is a function only of the public transcript $W,A,B,U,V$ only: there may be multiple  private operators
\[ u, \mu \in \mathcal{A} \slot \mathcal{B}, \quad v, \nu \in \mathcal{C} \slot \mathcal{D}  \]
such that 
\[ u(W)=U, \quad \mu(W)=U,   \quad  v(W)= V, \quad \nu(W) = V \]
(that is, multiple private choices for powers of $A$ and $B$ that produce the same public terms $U$ and $V$: multiple `witnesses' for the transcript) but since these operators all commute by construction, for each possible pairing of private choices $u(W)=\mu(W)=U$ and $v(W)=\nu(W)=V$ we get
\[   u(V) = \mu(V) = v(U) = \nu(U)  =K \] 
and  the shared key is  a function of the public transcript $W,A,B,U,V$ only.

\subsection{Shpilrain's extension}
Stickel's scheme was cryptanalysed by Shpilrain \cite{shpilrain2008cryptanalysis}, who noticed that to get the shared key, an eavesdropper Eve needs only to find witness matrices $a_e$ and $b_e$ such that
\[	a_e \in \Cen{A}, \quad  b_e \in \Cen{B}, \quad  a_e W b_e =U  \]
(and these witnesses need not be powers of $A$ and $B$).  Given such witnesses  Eve can then take
\[	a_e V b_e = a_e c  W d b_e =  c a_e   W  b_e  d = c U  d \]
giving the shared key.    The advantage in this attack is that any matrices $  a_e \in \Cen{A},  b_e \in \Cen{B}$ that reproduce $U$ will give the shared key; and if $U$ is invertible then all that is needed is to find $a_e \in \Cen{A}$ and invertible $x_e \in \Cen{B}$ satisfying the linear equation
\[    a_e W - U x_e = 0 \]
and taking $b_e= x_e^{-1}$ we have $a_e W b_e = U$ as required.

To avoid this linear algebra attack, Shpilrain  proposed a version using singular rather than invertible  matrices $A,B$ and using polynomials $f,g,\ell,m$ of $A$ and $B$ rather than powers \cite{shpilrain2005new}.  Letting  $Alg(X)$ be the polynomial algebra generated by $X$, in Shpilrain's scheme Alice has private matrices $a \in Alg(A),  b \in Alg(B)$ and Bob private matrices  $c \in Alg(A), d \in Alg(B)$, and Alice and Bob publish
\[ U= a W b, \quad  V =c  W d \]
Given $V$ Alice takes 
\[a  V b =a c W d b= a c W b d  \]
(since polynomials of a given element commute) and given $U$ Bob takes
\[c  U d =c a W b d= a c W b d   \]
getting a shared key   $K =  a c W b d $  as before.     In our notation Shpilrain's scheme is produced by defining  $\mathcal{A},\mathcal{C}$ both equal to $Alg(A)$ and  $\mathcal{B},\mathcal{D}$ both equal to $Alg(B)$, picking 
\[ u \in \mathcal{A} \slot \mathcal{B}, \quad v \in \mathcal{C} \slot \mathcal{D}  \] 
and running the scheme just as before; and since the operator structure of the scheme is unchanged, the shared key $K$ is again a function of the public transcript.

\subsection{Nager's extension}
We can also consider Nager's \cite{Nager2024} recent algebraic extension of Stickel's approach.  To construct a shared key in this scheme Alice picks  random integers $f,g,h,i$ and sets $a_1 =A^{f}, b_1 = B^{g}, a_2= A^{h}$ and $b_2=B^{i}$ and publishes the matrix
\[ U= a_1 W b_1 +  a_2 W b_2 \] 
while Bob  similarly picks  random integers $j,k,\ell,m$, sets $c_1=A^{j},d_1 = B^{k},c_2=A^{\ell},d_2=B^m$, and publishes  
\[ V= c_1 W d_1 +  c_2 W d_2 \] 
Given $V$ Alice takes 
\[a_1  V b_1 + a_2 V b_2 = a_1 c_1 W b_1 d_1  + a_1 c_2 W b_1 d_2 + a_2 c_1 W  b_2 d_1+  a_2 c_2 W b_2 d_2 \]
(the $a$ and $b$ terms commuting) and given $U$ Bob takes
\[c_1  U d_1 + c_2 U d_2 = a_1 c_1 W b_1 d_1  + a_1 c_2 W b_1 d_2 + a_2 c_1 W  b_2 d_1+  a_2 c_2 W b_2 d_2 \]
and Alice and Bob again have a shared key.  In our notation this scheme is produced by defining  $\mathcal{A},\mathcal{C}$ equal to $Alg(A)$, $\mathcal{B},\mathcal{D}$ equal to $Alg(B)$,  picking operators from  $\langle \mathcal{A} \slot \mathcal{B} \rangle$ and $\langle \mathcal{C} \slot \mathcal{D} \rangle$  and running the scheme just as before (with, again, the shared key $K$ being a function of the public transcript).

There are also a range of key exchange approaches which use this Stickel or two-sided multiplication approach over  tropical semirings  rather than finite fields.  Grigoriev and Shpilrain \cite{grigoriev2014tropical} gave the first such scheme, analogous to the Shpilrain polynomial scheme described above but using matrices over the tropical max-plus algebra; Kotov and Ushakov  \cite{kotov2018analysis} gave a cryptanalysis of this scheme.  A variety of other tropical Stickel-type schemes have been proposed,  using matrices of various forms to construct $a,c,b,d$ \cite[for  reviews of the tropical approach to cryptography, see][]{isaac2021closer,ahmed2023review}.  Since these differ from the above schemes in the choice of underlying representation rather than the key-exchange structure,  we do not analyse them in detail.

\subsection{A general form: commuting subspaces}  

We now describe a scheme that generalizes these different approaches. In the initialisation stage, Alice and Bob agree a public matrix $W$ and public matrix subspaces $\mathcal{A},\mathcal{B},\mathcal{C},\mathcal{D}$ such that  $\mathcal{C} \subseteq \Cen{\mathcal{A}},\mathcal{D} \subseteq  \Cen{\mathcal{B}}$  (the generators for the scheme). To form a shared key,  Alice picks some operator $u \in \langle \mathcal{A} \slot \mathcal{B} \rangle$ and publishes $ U = u (W)$ while Bob picks some  operator $v \in \langle \mathcal{C} \slot \mathcal{D} \rangle$ and publishes $ V =v (W)$
and since these operators commute they get their shared key by 
\[  u(V) =  u(v (W)) = v(u (W))  = v(U) = K  \]
This scheme includes the earlier schemes as special cases; and just as in those schemes, the shared key in this scheme is a function only of the public transcript $W,\mathcal{A},\mathcal{B},\mathcal{C},\mathcal{D},U,V$.

\subsection{Cryptanalysis of Stickel-type schemes} \label{sec:cryptanalysis}

Here we give a `witness-finding' cryptanalysis for the general commuting subspace scheme described above.  Alice and Bob's operators $u$ and $v$ in this scheme can be expressed as
\[ u(X) = \sum_{i=1}^m a_i X b_i, \quad  v(X) = \sum_{j=1}^{m'} c_j X d_j,\quad  \quad a_i \in \mathcal{A},b_i \in \mathcal{B},c_j \in \mathcal{C},d_j \in \mathcal{D} \] 
Using the Kronecker identity \cref{eq:kronecker}  these operators have corresponding  operator matrices
\[ F_u =  \sum_{i=1}^m b_i^\top \otimes a_i, \quad G_v =  \sum_{j=1}^{m'} d_j^\top \otimes c_j, \quad  \quad a_i \in \mathcal{A},b_i \in \mathcal{B},c_j \in \mathcal{C},d_j \in \mathcal{D} \]
such that
\[  u(X) = Y \longleftrightarrow F_u \Vec{X} = \Vec{Y}, \quad \quad v(X) = Y \longleftrightarrow G_v \Vec{X} = \Vec{Y} \]
and so
\[ F_u \Vec{W} = \Vec{U}, \quad G_v \Vec{W} = \Vec{V} \]
Since all such operators $u,v$ commute (by construction) the corresponding matrices necessarily also commute, giving 
\[  F_u \Vec{V} =   F_u G_v \Vec{W} =  G_v  F_u \Vec{W} =G_v \Vec{U} = \Vec{K} \]
 Letting
\[  A_1,\ldots, A_{\dim \mathcal{A}} \in \mathcal{A}, \quad B_1,\ldots, B_{\dim \mathcal{B}} \in \mathcal{B}  \] 
be basis matrices for the public subspaces $\mathcal{A}$ and $\mathcal{B}$ then every tensor product $b_i^\top \otimes a_i$ is a member of the linear subspace spanned by the set  of basis tensors
\[   \{  L_{ij} =  B_v^\top \otimes A_i \mid 1 \leq i \leq \dim \mathcal{A},\ 1 \leq j \leq \dim \mathcal{B}  \} \]
and so   the set of all possible admissible operator matrices $F$ is the set of all matrices
\[ F(\lambda) =  \sum_{i=1}^{\dim \mathcal{A}}\sum_{j=1}^{\dim \mathcal{B}} \lambda_{ij} L_{ij} \]
for values of the coefficient matrix $\lambda$ of dimensions $(\dim \mathcal{A}) \times (\dim \mathcal{B})$;
and every such matrix $F(\lambda)$ will necessarily commute with the matrix $G_v$ associated with Bob's operator $v$.   Eve can solve for $\lambda$ in the linear matrix equation
\[ F(\lambda) \Vec{W} =  \left(\sum_{i=1}^{\dim \mathcal{A}}\sum_{j=1}^{\dim \mathcal{B}} \lambda_{ij} L_{ij}  \right)  \Vec{W} =  \Vec{U}   \]
 in polynomial time via standard matrix algebra,  giving  $F(\lambda)\Vec{V}=\Vec{K}$ and recovering the shared  key.

Note that this  tensor-based analysis is a  generalisation of a  similar recent result by Otero Sanchez \cite{sanchez2025keyexchangeprotocolbased}, which gives a  cryptanalysis of key exchange schemes using two-sided actions on a semiring (focusing in particular on tropical semirings and digital semirings).

\section{Hidden commuting subspaces}

The central problem with the above schemes is that the spaces used to construct Alice and Bob's commuting operators $u$ and $v$ are public, and so these operators can be solved for linearly in terms of public tensor bases. Here we describe a two-sided multiplication scheme  where these commuting spaces are not publicly available in linear form, but are instead hidden via conjugation (by `transport').    Where the above schemes had two stages (initialization and key exchange) this hidden subspace approach has three: initialization (where Alice and Bob agree a set of public generators), subspace construction (where they exchange information allowing them to form appropriately commuting hidden subspaces $\mathcal{A},\mathcal{B},\mathcal{C},\mathcal{D}$), and key exchange (where they pick operators  $u \in \langle \mathcal{A} \slot \mathcal{B} \rangle$ and  $v \in \langle \mathcal{C} \slot \mathcal{D} \rangle$ respectively, publish $U=u(W),V=v(W)$, and get $K=u(V)=v(U)$ as before).  

For a given  $x$ and an invertible matrix $r$, we say that the conjugation $r x r^{-1}$ is the operation of transporting the core matrix $x$ by the transporter $r$;  any matrix we refer to as a `transporter' is implicitly invertible (and we typically take this as read).  In this scheme we use the fact that for a core subspace $\mathcal{X}$ and a transporter subspace $\mathcal{R}$, if Bob publishes a matrix
\[  \tilde{A} = r_b x_b r_b^{-1} \textit{ for } r_b \in \Cen{\mathcal{R}}^\times, x_b \in \mathcal{X}\]
and Alice publishes a matrix 
\[  \tilde{C} = r_a x_a r_a^{-1} \textit{ for } r_a \in \mathcal{R}^\times, x_a \in \Cen{\mathcal{X}}\]
then taking $r=r_a r_b=r_b r_a$ (since these transporters necessarily commute) Alice and Bob can form private matrices
\[  A = r_a \tilde{A} r_a^{-1} = r x_b r^{-1}, \quad \quad  C = r_b \tilde{C} r_b^{-1} = r x_a r^{-1} \]
and since $x_a$ and $x_b$ commute we have
\[ AC =   r x_b x_a r^{-1} =  r  x_a x_b r^{-1} = CA \]
and these two private matrices commute; and thus the private subspaces $\mathcal{A} = Alg(A),\mathcal{C}=Alg(C)$ also commute.

Given this, our proposed scheme runs as follows:

\paragraph{Initialisation:}
Alice and Bob agree a public invertible matrix $W$, two public core subspaces $\mathcal{X},\mathcal{Y}$ with  centralizers $\Cen{\mathcal{X}},\Cen{\mathcal{Y}}$, and two transporter subspaces $\mathcal{R},\mathcal{S}$ with centralizers $\Cen{\mathcal{R}},\Cen{\mathcal{S}}$ (where the transporter subspaces and their centralizers all contain invertible matrices).  These are the generators for the scheme. 

\paragraph{Subspace construction:} 
 Given these generators Alice constructs two  matrices $ \tilde{C},\tilde{D}$   such that
\begin{equation*}
	\tilde{C} = r_a x_a r_a^{-1},\tilde{D} = s_a y_a s_a^{-1}, \,\,  r_a \in \mathcal{R}^\times, x_a \in \Cen{\mathcal{X}}, s_a \in \Cen{\mathcal{S}}^\times, y_a \in \mathcal{Y}
\end{equation*}
for randomly chosen private transporters  $r_a,s_a$ and private core matrices $x_a,y_a$; she publishes $ \tilde{C}, \tilde{D}$  as her public keys, and  keeps $r_a,s_a$ as her private keys.  Similarly Bob  constructs two matrices $ \tilde{A},\tilde{B}$  such that
\begin{equation*} 
	\tilde{A} = r_b x_b r_b^{-1},  \tilde{B} = s_b y_b s_b^{-1}, \, \,  r_b \in \Cen{\mathcal{R}}^\times, x_b \in \mathcal{X}, s_b \in \mathcal{S}^\times, y_b \in  \Cen{\mathcal{Y}}
\end{equation*}
for  randomly chosen private transporters  $r_b,s_b$ and private core matrices $x_b,y_b$, publishes $ \tilde{A}, \tilde{B}$  as his public keys, and  keeps $r_b,s_b$ as his private keys.  

\paragraph{Key exchange:} 
 Alice now forms  private subspaces
\begin{equation*}
	 \mathcal{A} = r_a\, Alg(\tilde{A})\, r_a^{-1}, \quad  \quad \mathcal{B} =  s_a \, Alg(\tilde{B}) \, s_a^{-1} 
\end{equation*}
and  Bob forms  private subspaces
\begin{equation*} 
  \mathcal{C} = r_b\, Alg(\tilde{C})\, r_b^{-1}, \quad \quad  \mathcal{D} = s_b\,  Alg(\tilde{D}) \, s_b^{-1}
\end{equation*}
and given the commuting structure of the core and transporter subspaces, we have  $\mathcal{C} \subseteq \Cen{\mathcal{A}},\mathcal{D} \subseteq  \Cen{\mathcal{B}}$ as above.  

Now a general Stickel-type key exchange runs as before, but using these private commuting subspaces: Alice picks some operator $u \in \langle \mathcal{A} \slot \mathcal{B} \rangle$, Bob picks some operator $v \in \langle \mathcal{C} \slot \mathcal{D} \rangle$, they publish $ U = u (W),  V =v (W)$,
 Alice takes $K=u(V)=   u(v (W))$ and Bob takes $v(U) = v(u (W))  =  K$, and they have a shared key.

\subsection{Correctness and transcript dependence}

Assuming fixed generators  we represent a given key exchange transcript in this scheme as
\[
\tau = (\tilde{A},\tilde{B},\tilde{C},\tilde{D},U,V)
\]
For fixed generators,  Alice and Bob have some constrained set of  choices for their core and transporter matrices (those choices that are consistent with the requirements of the scheme);  for fixed $\tau$ we take
$\Omega_A(\tau) $ to be the set of honest choices $(x_a,y_a,r_a,s_a, u)$ for Alice for which
\[ r_a x_a r_a^{-1} = \tilde{C}, \quad s_a y_a s_a^{-1} = \tilde{D},  \quad 
 u(W)=U, \quad u \in \langle r_a Alg(\tilde{A}) r_a^{-1} \,  \slot\, s_a Alg(\tilde{B}) s_a^{-1} \rangle    \]
 also hold  (so $\Omega_A(\tau) $ is the set of all honest choices that also reproduce Alice's part of the transcript $\tau$).  We similarly take 
 $\Omega_B(\tau) $ to be the set of all honest choices $(x_b,y_b,r_b,s_b,v)$ for Bob for which
 \[ r_b x_b r_b^{-1} = \tilde{A}, \quad s_b y_b s_b^{-1} = \tilde{B},  \quad 
 v(W)=V, \quad v \in \langle r_b Alg(\tilde{C}) r_b^{-1}\,  \slot\, s_b Alg(\tilde{D}) s_b^{-1} \rangle    \]
 also hold (so $\Omega_B(\tau) $ is the set of all honest choices that also reproduce Bob's part of $\tau$).  
Using this notation we have 

\begin{proposition}\label{prop:transcript_key}
	For a fixed transcript $\tau$ generated in this hidden subspace scheme, all witness tuples in $\Omega_A(\tau) \times \Omega_B(\tau)$ yield the same shared key: the shared key is determined by $\tau$.
\end{proposition}
\begin{proof}
From the construction of this scheme and the fact that these sets  $\Omega_A(\tau), \Omega_B(\tau) $ contain only choices that are honestly consistent with the scheme's requirements and the transcript $\tau$, we see that for any pair of tuples
\[ (x_a,y_a,r_a,s_a, u) \in \Omega_A(\tau), \quad \quad (x_b,y_b,r_b,s_b, v)\in  \Omega_B(\tau)  \]
 we have $u(v(W))=v(u(W))$ (the operators in these tuples necessarily commute). 
	
 Given this, fix an operator $u_*$ for some tuple $(x_a,y_a,r_a,s_a, u_*) \in \Omega_A(\tau)$; then we have $u_*(W)= U$ for that operator (because this holds for all such operators), and can define a fixed value $K=u_*(V)$.  Then for any tuple in $\Omega_B(\tau) $ the operator $v$ in that tuple gives $v(W)=V$ (because this holds for all such operators) and $v$ commutes with the fixed operator $u_*$.    This means we have $v(U)=v(u_*(W))=u_*(v(W)) = u_*(V) = K$; and so for any tuple in $\Omega_B(\tau) $ the operator $v$ gives $v(U) = K$ for this fixed value $K$. 
 
  Since all pairs of operators  $u,v$ in these sets commute we have $u(V)=u(v(W))=v(u(W))=v(U)=K$ for all such pairs, and for any tuple in $\Omega_A(\tau) $ the operator $u$ also gives $u(V)=K$.
Thus all honest and transcript-consistent choices for Alice and Bob's private terms necessarily give the same shared key $K$ for that transcript: the shared key has a fixed value $K$ fully determined by the transcript $\tau$.
\end{proof}

\subsection{Parameter and subspace constructions}

This scheme gives Alice and Bob a shared key $K$ for all choices of prime  $p$ and matrix dimension $n$,   all choices for the public matrix $W$ and the public core  subspaces  $\mathcal{X},\mathcal{Y}$ with their centralizers $\Cen{\mathcal{X}},\Cen{\mathcal{Y}}$, and all choices for the transporter subspaces $\mathcal{R},\mathcal{S}$ and $\Cen{\mathcal{R}},\Cen{\mathcal{S}}$  where each transporter subspace contains invertible matrices.   For security we recommend large $n$ and some prime $p < n$, and core subspaces chosen so that $\mathcal{X},\Cen{\mathcal{X}},\mathcal{Y},\Cen{\mathcal{Y}}$ are all nontrivial and at least dimension $\alpha n$ for some constant $\alpha$.   We recommend transporter subspaces chosen so that $\mathcal{R}^\times,\Cen{\mathcal{R}}^\times,\mathcal{S}^\times,\Cen{\mathcal{S}}^\times$  are all nontrivial, do not commute with the core subspaces or with $W$, and so that all have dimensions of at least $\gamma n$ for some constant $\gamma$.   We also recommend that the core and transporter spaces $\mathcal{R},\mathcal{S},\mathcal{X},\mathcal{Y}$ be subject to conjugation by independent randomly chosen invertible matrices $\theta_R,\theta_S,\theta_X$ and $\theta_Y$, so any block or tensor structure is not shared across these spaces.

 A simple block construction giving transporter spaces of dimensions approximately $n^2/4$ (so exceeding these requirements) is to take $n=2k$, chose some invertible matrix $\theta_R$, and define
\[  \mathcal{R} =  \left\{  \theta_R \begin{bmatrix} r & 0 \\ 0 & z  \end{bmatrix} \theta_R^{-1} , \quad r \in \mathcal{M}_k, z \in \mathcal{Z}_k \right\}   \]
so that 
\[  \Cen{\mathcal{R}} =  \left\{  \theta_R \begin{bmatrix} z & 0 \\ 0 & r  \end{bmatrix} \theta_R^{-1} , \quad r \in \mathcal{M}_k, z \in \mathcal{Z}_k \right\}  \]
and both $\mathcal{R}$ and $ \Cen{\mathcal{R}}$ have dimensions $k^2+1 \approx n^2/4$ and contain invertible matrices;  a similar construction (with some conjugator $\theta_S$) gives $\mathcal{S}$ and $ \Cen{\mathcal{S}}$ with the same dimensions.    We don't require this transporter subspace construction: we simply give it to show an extreme example.

\section{Analysis} 
 
Just as in the various public subspace schemes described earlier, the shared key $K$ in this hidden subspace scheme is a function of the public  key-exchange transcript  only.   This  transcript, however, by design does not include the subspaces $\mathcal{A},\mathcal{B},\mathcal{C},\mathcal{D}$ over which  key exchange runs; and so our witness-finding cryptanalysis of public subspace key exchange in \cref{sec:cryptanalysis} does not apply directly to this scheme.  
  Instead, the natural path to key recovery here is one where an eavesdropper Eve identifies witness transporters that allow her to recover Alice or Bob's hidden subspaces; then can she get the shared key $K$ as in \cref{sec:cryptanalysis}.
The witness-finding form of this  approach matches Shpilrain's  attack on Stickel's original scheme, and our cryptanalysis of generalized Stickel-type schemes; as far as we know, the same witness-finding pattern also covers all successful attacks on tropical Stickel-type variants:  our heuristic security assumption is that key exchange in this scheme will be secure if it is secure against such a witness-finding attack.
Here we give an initial analysis of the complexity of this witness-finding approach, showing that witness-finding key recovery is NP-hard in the worst case (via reduction to a standard problem).   Note that we do not have any proof of hardness for the scheme in general: there may be approaches that avoid witness-finding entirely and recover the shared key some other way. 

To recover the shared key in this scheme via witness-finding, Eve solves for matrices $x,y,\rho,\sigma$ such that
\[ \rho^{-1} \tilde{C} \rho = x, \quad x \in \Cen{\mathcal{X}}, \rho \in \mathcal{R}^\times, \quad \sigma^{-1} \tilde{D} \sigma = y, \quad y\in \mathcal{Y},  \sigma \in \Cen{\mathcal{S}}^\times  \]
and if she finds an operator $\mu$ such that
\[   \mu(W)=U, \quad \mu \in \langle \mathcal{A}_e \,  \slot\, \mathcal{B}_e \rangle \textit{ where }   \mathcal{A}_e = \rho Alg(\tilde{A}) \rho^{-1},\, \mathcal{B}_e =  \sigma Alg(\tilde{B})\sigma^{-1}  \]
then she has a set of witnesses $(x,y,\rho,\sigma,\mu) \in \Omega_A(\tau)$ and gets the shared key by taking $\mu(V)=K$.    Given such transporters $\rho,\sigma$, solving for such an operator $\mu$ can be done via linear algebra over tensors as in \cref{sec:cryptanalysis}; and so the  computational challenge for Eve is to solve the following problem: \\

\fbox{%
	\parbox{0.9\linewidth }
	{ \textbf{Problem 1:} Restricted conjugacy into subspaces. \\
		
		Given a matrix $Q$, a subspace $\mathcal{T}$, and a subspace  $\Phi$  (all in $\mathcal{M}$) return a matrix $\rho$ such that 
		\[\rho \in \Phi^{\times}, \, \, \rho^{-1} Q \rho = t \textit{ for some } t \in \mathcal{T} \] 
		or report failure if no such $\rho$ exists.
		\\ 
}}

\vspace{12pt}

This is close to several well-studied orbit and matrix-space problems.  If the target matrix $t$  is fixed rather than ranging over the subspace $\mathcal{T}$,  Problem $1$  reduces to a  matrix-tuple conjugacy or transporter problem: given two fixed matrices $Q,t$, this asks for a transporter $\rho$ such that $\rho^{-1} Q \rho = t$, with recent work of Qiao and Sun providing polynomial-time canonical form solutions for matrix tuples over finite fields  \cite{qiao2024canonical}.   The extra features in our setting are that the target matrix $t$ is not fixed but is an unknown member of the subspace $\mathcal{T}$, and that the transporter $\rho$ is itself restricted to lie in a specified subspace $\Phi$.  With these features we can show that Problem $1$ has a reduction from the standard hardness problem  \\

\fbox{%
	\parbox{0.9\linewidth }
	{ \textbf{Problem 2:} Edmonds' problem (search version) \\
		
	  Given a subspace $\mathcal{B} \subseteq \mathcal{M}_{k}$, return a matrix  $b \in \mathcal{B}^\times$ (an invertible member of $\mathcal{B}$) or report failure if no such invertible member exists.
		\\ 
}} 
\vspace{12pt}

We show this with a block construction that uses the same block-scalar transporter shape as the illustrative transporter spaces described above. Let $\mathcal B\leq \mathcal M_k$ be an instance of Edmonds' problem. Set $n=2k$  and define
\[
\Phi=
\left\{
	\begin{bmatrix}
		\phi &0\\
		0&z
	\end{bmatrix}, 	
	\phi \in\mathcal M_k,\ z\in\mathcal Z_k
	\right\}.
\]
Choose a random invertible matrix $q\in\mathcal M_k^\times$ and random matrices $Q_{11},Q_{21},Q_{22} \in\mathcal M_k$ and set
\[
Q=
\begin{bmatrix}
	Q_{11} &q\\
	Q_{21} & Q_{22}
\end{bmatrix}
\]
and define the target subspace
\[
\mathcal T_{\mathcal B}=
\left\{
	\begin{bmatrix}
		T_{11}&b\\
		T_{21}&T_{22}
	\end{bmatrix}
	:
	T_{11},T_{21},T_{22}\in\mathcal M_k,\ b\in\mathcal B
	\right\}.\]
(a linear subspace of $\mathcal M_{n}$)

Now take
\[
\rho=
\begin{bmatrix}
	\phi &0\\
	0&z
\end{bmatrix}
\in\Phi^\times,
\qquad
\phi \in\mathcal M_k^\times,\ z\in\mathcal Z_k^\times.
\]
and we have
\[
\rho^{-1}Q\rho =
\begin{bmatrix}
	\phi ^{-1}&0\\
	0&z^{-1}
\end{bmatrix}
\begin{bmatrix}
	Q_{11} &q\\
	Q_{21} & Q_{22}
\end{bmatrix}
\begin{bmatrix}
	\phi &0\\
	0&z
\end{bmatrix}
=
\begin{bmatrix}
\phi ^{-1} Q_{11} \phi  &\phi ^{-1}qz\\
	z^{-1} Q_{11}  \phi   &  Q_{22}  
\end{bmatrix}.
\]
and $\rho Q \rho^{-1} \in \mathcal{T}_\mathcal{B}$ when $\phi ^{-1}qz \in \mathcal{B}$ (all other requirements for membership in $\mathcal{T}_\mathcal{B}$ being satisfied automatically).
Since $\phi ,q,z$ are invertible this means that any solution $\rho$ to Problem $1$ gives an invertible element
\[
b=\phi ^{-1}qz\in\mathcal B^\times,
\]
and hence  a solution to Problem $2$.

Conversely, suppose $\mathcal B$ contains an invertible element $b\in\mathcal B^\times$ (a solution to Problem $2$). Taking $z=I$ and
\[
\phi =qb^{-1}
\]
gives $\phi \in\mathcal M_k^\times$ and
\[
\phi ^{-1}qz=\phi ^{-1}q=b.
\]
Thus any solution $b$ to Problem $2$ gives a transporter $\rho$ such that
\[
\rho=
\begin{bmatrix}
	qb^{-1}&0\\
	0&I
\end{bmatrix}
\in\Phi^\times, \quad \quad 
\rho^{-1}Q\rho\in\mathcal T_{\mathcal B}.
\]
and hence  gives a solution to Problem $1$.   This means that every instance of Problem $2$  can be solved by obtaining a solution to a corresponding instance of  Problem $1$; and so Problem $1$ must be at least as hard as Problem $2$.  Problem $2$ is known to be NP-hard in the matrix dimension $k$ for fixed field $\mathbb{F}_p$ \cite{buss1999computational}, though it does have polynomial time randomised solutions for  $p > 2k =n$ \cite{lovasz1979determinants,ivanyos2015generalized}, and so solving Problem $1$ will be  NP-hard in the worst case, at least when $p < n$; and this witness-finding approach to key recovery inherits this worst-case hardness.
  
 \section{Discussion}
 
 We've given a general cryptanalysis showing that  Stickel-type key exchange schemes over public subspaces can be broken in polynomial time by linear algebra over tensor products of those subspaces, and given a new version of this type of scheme where the key-exchange subspaces are not public, blocking this attack.  We have also shown that to recover the shared key in these schemes via witness-finding (the standard approach for Stickel-type schemes) in the worst case requires an eavesdropper Eve to solve an NP-hard problem.  This does not, of course, show that finding the shared key in this scheme is unconditionally NP-hard: there may be approaches that avoid witness-finding entirely and recover the shared key in polynomial time some other way. 
 
 One possible drawback of the hidden subspaces approach described here is that it gives our eavesdropper Eve public subspaces containing all private core and transporter matrices used in the scheme (and so linear algebraic structure that can be used to express those private matrices).  In a second paper we give a more complex version of this hidden subspace approach, where transporters are formed via local rather than global commutation so that, for example, Alice and Bob's transporters $r_a$ and $r_b$ do not themselves commute but where
 \[  r_b r_a\, x\, r_a^{-1} r_b^{-1} =  r_a r_b\, x\, r_b^{-1} r_a^{-1}   \]
 holds for all $x \in \mathcal{X}$ (the transporters commute locally at $\mathcal{X}$).   Alice and Bob's transporters are not members of public subspaces $\mathcal{R},\mathcal{S}$ (or their centralizers) in this `local commuting transporters' scheme, but instead are solutions to public non-linear problems (versions of Problem $1$).  This means that no linear algebraic structure describing those transporters is  available to Eve, and getting the shared key by recovery of witness transporters $\rho_a,\sigma_a$ becomes a more complex problem.
  
\bibliographystyle{unsrt}
\bibliography{references}

@preamble{ " \newcommand{\noop}[1]{} " }

@inproceedings{qiao2024canonical,
	author    = {Qiao, Youming and Sun, Xiaorui},
	title     = {Canonical Forms for Matrix Tuples in Polynomial Time},
	booktitle = {Proceedings of the 65th IEEE Symposium on Foundations of Computer Science (FOCS)},
	year      = {2024},
	note      = {Available as arXiv:2409.12457}
}

@article{buss1999computational,
	title={The computational complexity of some problems of linear algebra},
	author={Buss, Jonathan F and Frandsen, Gudmund S and Shallit, Jeffrey O},
	journal={Journal of Computer and System Sciences},
	volume={58},
	number={3},
	pages={572--596},
	year={1999},
	publisher={Elsevier}
}

@inproceedings{lovasz1979determinants,
	title={On determinants, matchings, and random algorithms.},
	author={Lov{\'a}sz, L{\'a}szl{\'o}},
	booktitle={FCT},
	volume={79},
	pages={565--574},
	year={1979}
}

@article{ivanyos2015generalized,
	title={Generalized Wong sequences and their applications to Edmonds' problems},
	author={Ivanyos, G{\'a}bor and Karpinski, Marek and Qiao, Youming and Santha, Miklos},
	journal={Journal of Computer and System Sciences},
	volume={81},
	number={7},
	pages={1373--1386},
	year={2015},
	publisher={Elsevier}
}

@misc{sanchez2025keyexchangeprotocolbased,
	title={On key exchange protocol based on Two-side multiplication action}, 
	author={Alvaro Otero Sanchez},
	year={2025},
	eprint={2504.15880},
	archivePrefix={arXiv},
	primaryClass={cs.CR},
	url={https://arxiv.org/abs/2504.15880}, 
}

@article{shpilrain2005new,
	title={A new key exchange protocol based on the decomposition problem},
	author={Shpilrain, Vladimir and Ushakov, Alexander},
	journal={arXiv preprint math/0512140},
	year={2005}
}

@article{grigoriev2014tropical,
	title={Tropical cryptography},
	author={Grigoriev, Dima and Shpilrain, Vladimir},
	journal={Communications in Algebra},
	volume={42},
	number={6},
	pages={2624--2632},
	year={2014},
	publisher={Taylor \& Francis}
}

@article{isaac2021closer,
	title={A closer look at the tropical cryptography},
	author={Isaac, Steve and Kahrobaei, Delaram},
	journal={International Journal of Computer Mathematics: Computer Systems Theory},
	volume={6},
	number={2},
	pages={137--142},
	year={2021},
	publisher={Taylor \& Francis}
}

@article{ahmed2023review,
	title={A review of the tropical approach in cryptography},
	author={Ahmed, Kashif and Pal, S and Mohan, Radha},
	journal={Cryptologia},
	volume={47},
	number={1},
	pages={63--87},
	year={2023},
	publisher={Taylor \& Francis}
}

@misc{Nager2024,
	author = {Daniel Nager},
	title = {{S}tickel's Key Agreement Algebraic Variation},
	howpublished = {Cryptology {ePrint} Archive, Paper 2024/792},
	year = {2024},
	url = {https://eprint.iacr.org/2024/792}
}

@article{kotov2018analysis,
	title={Analysis of a key exchange protocol based on tropical matrix algebra},
	author={Kotov, Matvei and Ushakov, Alexander},
	journal={Journal of Mathematical Cryptology},
	volume={12},
	number={3},
	pages={137--141},
	year={2018}
}

@inproceedings{stickel2005new,
	title={A new method for exchanging secret keys},
	author={Stickel, Eberhard},
	booktitle={Third International Conference on Information Technology and Applications (ICITA'05)},
	volume={2},
	pages={426--430},
	year={2005},
	organization={IEEE}
}

@inproceedings{shpilrain2008cryptanalysis,
	title={Cryptanalysis of {S}tickel's key exchange scheme},
	author={Shpilrain, Vladimir},
	booktitle={International computer science symposium in Russia},
	pages={283--288},
	year={2008},
	organization={Springer}
}
 
\end{document}